\documentclass[conference]{IEEEtran}

\usepackage[letterpaper,left=0.680in,right=0.680in,bottom=0.990in,top=0.71in]{geometry}
\usepackage{ifpdf}
\ifCLASSINFOpdf
  \usepackage[pdftex]{graphicx}
  \graphicspath{{./figures/}}
  \DeclareGraphicsExtensions{.pdf,.png,.jpeg}
\else
  \usepackage[dvips]{graphicx}
  \DeclareGraphicsExtensions{.eps}
\fi
\usepackage{amsmath,amssymb,amsfonts}
\usepackage{amsthm}
\usepackage{booktabs}
\usepackage{array}
\usepackage{cite}
\usepackage{url}
\usepackage[caption=false,font=footnotesize]{subfig}
\usepackage{tikz}
\usetikzlibrary{arrows.meta,positioning,calc,shapes.geometric}

\newtheorem{proposition}{Proposition}
\newtheorem{lemma}{Lemma}
\newtheorem{remark}{Remark}

\begin{document}

\title{Fourth-Order Co-Channel Interference-Aware OFDM-ISAC Sensing in Cluttered Environments}

\author{
  \IEEEauthorblockN{Tri~Nhu~Do, Yosefine~Triwidyastuti, and Gunes~Karabulut~Kurt}
  \IEEEauthorblockA{Department of Electrical Engineering, Polytechnique Montr\'eal, Montr\'eal, QC, Canada\\
  \{tri-nhu.do, yosefine-2.triwidyastuti, gunes.kurt\}@polymtl.ca}
}

\maketitle

\begin{abstract}
In this paper, we develop a sensing-centric receiver for monostatic orthogonal frequency-division multiplexing integrated sensing and communication (OFDM-ISAC) in the simultaneous presence of passive clutter and a non-cooperative co-channel OFDM transmitter. We show that removing the known ISAC symbols maps target and clutter echoes to deterministic delay-Doppler components, while the non-cooperative waveform remains a random, non-Gaussian, generically full-rank residual. Guided by this distinction, we combine moving-target indication for zero-Doppler clutter with a centered diagonal fourth-order statistic for interference identification. To obtain an unbiased fourth-order estimate from phase-aligned proper snapshots, we derive additive and multiplicative finite-sample bias corrections and pair the corrected statistic with a second-order moment to separate interference from Gaussian noise. Fourth order provides identification and presence testing, whereas second order provides the diagonal disturbance-power proxy used for weighting. We further introduce diagonally loaded weights that enforce a range-Doppler aperture-efficiency floor. Simulations show that the interference estimate approaches its theoretical profile as coherent support grows and remains insensitive to deterministic echoes, that the presence test resolves interference well below the level at which weighting matters, and that loaded processing improves target detection while controlling point-spread-function degradation.
\end{abstract}


\section{Introduction}
\label{sec:intro}

Integrated sensing and communication (ISAC) reuses spectrum, radio hardware, and waveforms for data transmission and environmental sensing \cite{Triwidyastuti2026arxiv}. Orthogonal frequency-division multiplexing (OFDM) is particularly attractive because subcarrier and symbol indices encode delay and Doppler, respectively, and because a colocated sensing receiver knows the transmitted data \cite{MirabellaAccess2023,Thinh_TGCN_2026,Richards2014}. A monostatic transceiver can therefore form a two-dimensional range-Doppler image directly from its time-frequency grid. Most OFDM sensing models, however, emphasize additive noise or deterministic echoes. In practice, a monostatic receiver operates in a complex sensing environment with strong passive clutter and co-channel radiation from a non-cooperative transmitter \cite{MengTWC2024,Nguyen_TWC_arXiv_2026}. Clutter-aware ISAC processing has recently received attention \cite{LuoTWC2024}, while interference is often absorbed into an effective noise floor. That abstraction loses an important distinction: clutter is a reflection of the known waveform, but a co-channel interferer carries unknown data. Dividing the observation by the local reference symbols therefore affects the two disturbances differently.

OFDM studies have established waveform tradeoffs, delay-Doppler processing, and high-dimensional channel estimation \cite{MirabellaAccess2023,Thinh_TGCN_2026,Richards2014}, while clutter-aware ISAC uses spatial, Doppler, or mean-phasor cancellation \cite{LuoTWC2024}. Here the desired echoes still become complex exponentials, but a non-cooperative transmitter remains randomly modulated and violates that source model. Higher-order statistics distinguish this non-Gaussian residual from Gaussian noise \cite{MendelProcIEEE1991}; they are not used to create cumulant-bearing target steering vectors.
Classical higher-order statistics exploit the vanishing of cumulants above order two for Gaussian processes \cite{MendelProcIEEE1991}. Unlike blind source-separation formulations, however, the useful echoes here are deterministic and cumulant-free, while the only non-Gaussian component is outside the delay-Doppler manifold. This makes a cellwise statistic, rather than a tensor factorization, the appropriate construction.

In this paper, we exploit that difference from a sensing-centric perspective. We retain the conventional OFDM communication chain and focus on the base station (BS) sensing receiver. Our contributions are as follows.
\begin{itemize}
  \item We show that reference removal maps target and clutter echoes to deterministic low-rank delay-Doppler components, but maps a non-cooperative OFDM interferer to a random, non-Gaussian, generically full-rank residual.
  \item We propose a scalable diagonal fourth-order statistic that rejects deterministic echoes and Gaussian noise and supports interference-presence testing without prior knowledge of the noise level.
  \item We show that centering biases the statistic both additively and multiplicatively, derive both corrections in closed form to obtain an estimator that is unbiased for any proper disturbance when $M\geq4$, and combine it with a second-order moment to separate the interference and noise profiles.
  \item We distinguish interference identification from suppression and, under a diagonal post-MTI covariance approximation, propose loaded disturbance weighting that balances coherent gain and aperture efficiency while limiting point-spread-function degradation.
\end{itemize}

\textit{Notation}: $\vec{a}$ and $\mathbf{A}$ denote a vector and a matrix; $(\cdot)^{*}$, $(\cdot)^{T}$, and $(\cdot)^{H}$ are conjugation, transpose, and Hermitian transpose; $\odot$ is the Hadamard product; $\mathbf{I}_N$ and $\vec{1}_N$ are the identity matrix and all-ones vector of size $N$; $\operatorname{tr}(\cdot)$, $\|\cdot\|$, and $\lceil\cdot\rceil$ are the trace, Euclidean norm, and ceiling; $\mathbb{E}[\cdot]$ is expectation, $\operatorname{cum}_4(\cdot)$ the fourth-order cumulant, and $[x]_+=\max(x,0)$.

\section{System Model and Structural Distinction}
\label{sec:model}

\subsection{Topology and Propagation Geometry}

Consider a monostatic BS, a served user equipment (UE), moving targets, passive clutter, and a non-cooperative co-channel OFDM transmitter. Here, $\mathsf B$ is the ISAC BS, $\mathsf U$ is its downlink UE, $\mathsf I$ is a non-cooperative co-channel OFDM transmitter, that is, an uncoordinated source whose data symbols are unknown to the BS, $\mathcal T=\{\mathsf T_\ell\}_{\ell=1}^{L}$ contains moving targets, and $\mathcal C=\{\mathsf C_g\}_{g=1}^{G}$ contains passive clutter-generating scatterers. The links $\mathsf B\!\to\!\mathsf U$, $\mathsf I\!\to\!\mathsf B$, and $\mathsf I\!\to\!\mathsf U$ are one-way, whereas the sensing paths $\mathsf B\!\to\!\mathsf T_\ell\!\to\!\mathsf B$ and $\mathsf B\!\to\!\mathsf C_g\!\to\!\mathsf B$ are two-way. Thus, $\mathsf I$ is an active source carrying unknown symbols, whereas targets and clutter return the BS waveform. Fig.~\ref{fig:topology} illustrates the node and link types.

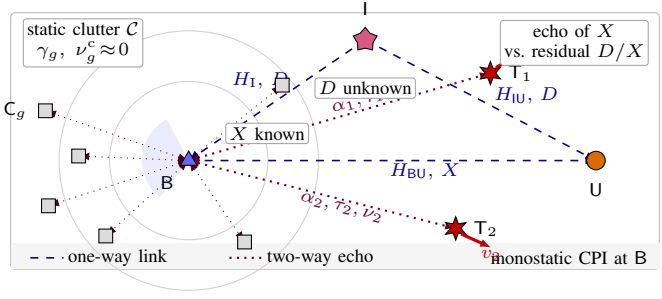
\begin{figure}[t]
\centering
\resizebox{\columnwidth}{!}{%
\begin{tikzpicture}[
        >={Latex[length=1.4mm]},
        font=\scriptsize,
        comm/.style={dashed,semithick,blue!55!black},
        sens/.style={dotted,thick,purple!55!black},
        sensc/.style={dotted,thin,purple!40!black},
        bsN/.style={regular polygon,regular polygon sides=3,draw=black,
                    fill=blue!60!white,minimum size=7.2pt,inner sep=0pt},
        intN/.style={star,star points=5,draw=black,fill=purple!65,
                     minimum size=9.4pt,inner sep=0pt},
        ueN/.style={circle,draw=black,fill=orange!85!black,
                    minimum size=7.0pt,inner sep=0pt},
        tgtN/.style={star,star points=6,star point ratio=2.2,draw=black,
                     fill=red!75!black,minimum size=8.8pt,inner sep=0pt},
        clN/.style={rectangle,draw=black,fill=black!14,
                    minimum size=5.2pt,inner sep=0pt},
        lbl/.style={font=\scriptsize,inner sep=0.6pt},
        tag/.style={draw=black!35,rounded corners=1.1pt,fill=white,
                    font=\scriptsize,inner sep=1.4pt,align=center}
    ]
        \draw[black!28,rounded corners=1.6pt]
              (-0.18,-1.22) rectangle (8.42,2.18);
        \fill[black!4] (-0.18,-1.22) rectangle (8.42,-0.88);

        \coordinate (B)  at (2.18,0.22);
        \coordinate (U)  at (7.58,0.22);
        \coordinate (I)  at (4.52,1.82);
        \coordinate (T1) at (6.18,1.38);
        \coordinate (T2) at (5.72,-0.68);
        \coordinate (C1) at (0.28,0.88);
        \coordinate (C2) at (0.72,0.28);
        \coordinate (C3) at (0.32,-0.38);
        \coordinate (C4) at (1.08,-0.78);
        \coordinate (C5) at (2.92,-0.86);
        \coordinate (C6) at (3.42,1.22);

        \foreach \r in {1.05,1.72}
            \draw[black!18] (B) circle (\r);
        \fill[blue!8] (B) -- ++(118:0.62) arc (118:232:0.62) -- cycle;

        \foreach \C in {C1,C2,C3,C4,C5,C6}
            \draw[sensc,<->] (B) -- (\C);

        \draw[comm,->] (B) -- node[lbl,below,pos=.58] {$H_{\mathsf{BU}},\,X$} (U);
        \draw[comm,->] (I) -- node[lbl,above left,pos=.42] {$H_{\rm I},\,D$} (B);
        \draw[comm,->] (I) -- node[lbl,above right,pos=.55] {$H_{\mathsf{IU}},\,D$} (U);
        \draw[sens,<->] (B) -- node[lbl,sloped,above,pos=.62]
              {$\alpha_1,\tau_1,\nu_1$} (T1);
        \draw[sens,<->] (B) -- node[lbl,sloped,below,pos=.58]
              {$\alpha_2,\tau_2,\nu_2$} (T2);

        \node[clN] (C1n) at (C1) {};
        \node[clN] (C2n) at (C2) {};
        \node[clN] (C3n) at (C3) {};
        \node[clN] (C4n) at (C4) {};
        \node[clN] (C5n) at (C5) {};
        \node[clN] (C6n) at (C6) {};
        \node[bsN]  (Bn)  at (B)  {};
        \node[intN] (In)  at (I)  {};
        \node[ueN]  (Un)  at (U)  {};
        \node[tgtN] (T1n) at (T1) {};
        \node[tgtN] (T2n) at (T2) {};

        \node[below left=0.4pt and 0.2pt of Bn] {$\mathsf B$};
        \node[below=1.2pt of Un] {$\mathsf U$};
        \node[above=1.0pt of In] {$\mathsf I$};
        \node[right=1.1pt of T1n] {$\mathsf T_1$};
        \node[right=1.1pt of T2n] {$\mathsf T_2$};
        \node[tag,anchor=west] at (-0.02,1.78)
              {static clutter $\mathcal C$\\[-0.4pt]
               $\gamma_g,~\nu_g^{\rm c}\!\approx\!0$};
        \node[left=0.6pt of C1n,font=\scriptsize] {$\mathsf C_g$};

        \draw[-{Latex[length=1.5mm]},very thick,red!75!black]
              (T1n.north east) -- ++(0.40,0.22)
              node[above,font=\scriptsize,inner sep=0.4pt] {$v_1$};
        \draw[-{Latex[length=1.5mm]},very thick,red!75!black]
              (T2n.south east) -- ++(0.40,-0.18)
              node[below,font=\scriptsize,inner sep=0.4pt] {$v_2$};

        \node[tag] at (3.22, 0.58) {$X$ known};
        \node[tag] at (4.52, 1.18) {$D$ unknown};
        \node[tag,anchor=east] at (8.28,1.78)
              {echo of $X$\\[-0.4pt] vs.\ residual $D/X$};

        \draw[comm] (0.08,-1.07) -- (0.52,-1.07);
        \node[anchor=west,font=\scriptsize,inner sep=0.4pt] at (0.56,-1.07)
              {one-way link};
        \draw[sens] (2.72,-1.07) -- (3.16,-1.07);
        \node[anchor=west,font=\scriptsize,inner sep=0.4pt] at (3.20,-1.07)
              {two-way echo};
        \node[anchor=east,font=\scriptsize,inner sep=0.4pt] at (8.28,-1.07)
              {monostatic CPI at $\mathsf B$};
    \end{tikzpicture}%
}
\caption{System topology: one-way communication/interference links and two-way sensing paths, with multiple static clutter scatterers and two moving targets.}
\label{fig:topology}
\end{figure}

Let $\vec{p}_a$ denote the position of node $a$. Target $\ell$ has $d_\ell=\|\vec{p}_{\mathsf T_\ell}-\vec{p}_{\mathsf B}\|$, radial velocity $v_\ell$ (positive when receding), round-trip delay, and monostatic Doppler
\begin{equation}
\tau_\ell=2d_\ell/c,\qquad
\nu_\ell=-2v_\ell/\lambda.                                  \label{eq:geometry}
\end{equation}
Target $\ell$ has two-way amplitude $\alpha_\ell$ \cite{Triwidyastuti2026arxiv}, and clutter component $g$ is described by $(\gamma_g,\tau_g^{\rm c},\nu_g^{\rm c})$, with $|\nu_g^{\rm c}|\ll1/(NT_{\rm s})$ when static. These coefficients are deterministic over a coherent processing interval (CPI).

\subsection{OFDM Resource Grid and Transmitter}

The BS transmits $K$ subcarriers over $N$ OFDM symbols per CPI. With $\Delta f$, $T=1/\Delta f$, $T_{\rm s}=T+T_{\rm cp}$, and alphabet $\mathcal A$, the resource element satisfies $X[k,n]\in\mathcal A$, $\mathbb{E}\big[|X[k,n]|^2\big]=1$. Its unitary $K$-point inverse fast Fourier transform (IFFT) is $x[m,n]=(1/\sqrt K)\sum_{k=0}^{K-1}X[k,n]e^{j2\pi km/K}$ for $0\leq m<K$. Let $T_{\rm samp}=T/K$ and $N_{\rm cp}=T_{\rm cp}/T_{\rm samp}$. Cyclic-prefix (CP) extension gives $x_{\rm cp}[m,n]=x([m]_K,n)$ for $-N_{\rm cp}\leq m<K$. The transmitted complex envelope $s_{\rm B}(t)$ serializes $x_{\rm cp}[m,n]$ at interval $T_{\rm samp}$, scaled so that $P_t$ is the average useful-symbol transmit power. The non-cooperative transmitter similarly radiates independent unit-power symbols $D[k,n]$ unknown to the BS and UE.

\subsection{UE Communication Receiver}

After radio-frequency (RF) downconversion, analog-to-digital converter (ADC) sampling, synchronization, CP removal, and a $K$-point fast Fourier transform (FFT), the UE observes $Y_{\mathsf U}[k,n]=\sqrt{P_t}H_{\mathsf BU}[k,n]X[k,n]+H_{\mathsf IU}[k,n]D[k,n]+W_{\mathsf U}[k,n]$.
On pilot set $\mathcal P$, least squares followed by interpolation $\mathcal I_{\mathcal P}\{\cdot\}$ gives $\widehat H_{\mathsf BU}=\mathcal I_{\mathcal P}\{Y_{\mathsf U}/(\sqrt{P_t}X)\}$ \cite{Thinh_TGCN_2026}. One-tap equalization, constellation detection, and demapping are
\begin{align}
\widehat X_{\mathsf U}&=\arg\min_{a\in\mathcal A}
|\widetilde X_{\mathsf U}-a|^2,\qquad
\widehat{\vec{b}}=\mu^{-1}(\widehat X_{\mathsf U}).                    \label{eq:uechain}
\end{align}
where $\widetilde X_{\mathsf U}=\frac{Y_{\mathsf U}}
{\sqrt{P_t}\widehat H_{\mathsf BU}}$, $\mu$ is the symbol mapper; indices $(k,n)$ are suppressed in (\ref{eq:uechain}). Unlike the BS, the UE does not know the entire frame before detection and cannot perform sensing-side reference removal.

This paper is sensing-centric and adopts this chain without modification to establish the transmitted reference and served-UE link. It does not propose a new waveform or UE communication receiver; its contribution begins with sensing-side reference removal at the BS.

\subsection{Sensing Receiver and Frequency-Domain Model}

After RF downconversion, the BS complex-baseband signal before sampling is $r_{\mathsf B}(t)=\sum_{\ell=1}^{L}\alpha_\ell s_{\mathsf B}(t-\tau_\ell)e^{j2\pi\nu_\ell t}+\sum_{g=1}^{G}\gamma_gs_{\mathsf B}(t-\tau_g^{\rm c})e^{j2\pi\nu_g^{\rm c}t}+(h_{\mathsf IB}*d_{\mathsf I})(t)+w_{\mathsf B}(t)$, where $d_{\mathsf I}(t)$ is the interferer complex envelope and $*$ denotes convolution. ADC sampling, synchronization, CP removal, and the FFT produce
\begin{align}
Y[k,n]={}&\sqrt{P_t}X[k,n]\big(H_{\rm t}[k,n]+H_{\rm c}[k,n]\big)
 \nonumber\\[-1mm]
&+H_{\rm I}[k,n]D[k,n]+W[k,n],                                      \label{eq:y}
\end{align}
where $W[k,n]\sim\mathcal{CN}(0,\sigma_w^2)$ is Gaussian noise and $D[k,n]$ is an independent unit-power interferer symbol. The target response is $H_{\rm t}[k,n]=\allowbreak\sum_{\ell=1}^{L}\alpha_\ell\allowbreak e^{-j2\pi k\Delta f\tau_\ell}\allowbreak e^{j2\pi\nu_\ell nT_{\rm s}}$, while the clutter response is $H_{\rm c}[k,n]=\allowbreak\sum_{g=1}^{G}\gamma_g\allowbreak e^{-j2\pi k\Delta f\tau_g^{\rm c}}\allowbreak e^{j2\pi\nu_g^{\rm c}nT_{\rm s}}$. The interference channel may be frequency and time selective, with $H_{\rm I}[k,n]=\allowbreak\sum_{r=1}^{R_{\rm I}}\varrho_r\allowbreak e^{-j2\pi k\Delta f\vartheta_r}\allowbreak e^{j2\pi\varphi_rnT_{\rm s}}$.
Here $(\vartheta_r,\varphi_r)$ are one-way interference-channel parameters, unlike the two-way target $(\tau_\ell,\nu_\ell)$ in (\ref{eq:geometry}).
We assume common OFDM numerology, timing within the CP, all delay spreads below $T_{\rm cp}$, and negligible Doppler inter-carrier interference (ICI) and range migration; the interferer power is absorbed into $H_{\rm I}$.

Because the BS knows the full transmitted frame, it removes the local reference:
\begin{align}
	\begin{aligned}
		&Z[k,n]=\frac{Y[k,n]}{\sqrt{P_t}X[k,n]} \\
		&=H_{\rm t}[k,n]+H_{\rm c}[k,n]
		+\frac{H_{\rm I}[k,n]}{\sqrt{P_t}}U[k,n]+\widetilde W[k,n]. \label{eq:z}
	\end{aligned}
\end{align}
Here $U[k,n]\triangleq\allowbreak D[k,n]/X[k,n]$, while $\widetilde W[k,n]\triangleq\allowbreak W[k,n]/(\sqrt{P_t}X[k,n])$. Equation~(\ref{eq:z}) performs reference normalization using the known transmitted symbol rather than conventional zero-forcing channel estimation. Since the transmitted sensing symbols are known at the monostatic BS, direct normalization removes the local reference exactly. Under the constant-modulus reference adopted below, $1/X[k,n]=X^*[k,n]$, so the operation reduces to phase de-rotation and introduces no symbol-dependent noise enhancement. A minimum mean-square error (MMSE) operation would require prior second-order models for the echo and disturbance components and would no longer provide the exact structural separation required by the subsequent cumulant analysis. Reference removal exactly cancels $X$ from every echo but not from the non-cooperative waveform. The resulting target or clutter matrix is a sum of outer products
\begin{equation}
\mathbf{Z}_{\rm e}=\sum_q c_q\vec{a}_f(\tau_q)\vec{a}_t^T(\nu_q),             \label{eq:rankone}
\end{equation}
where $[\vec{a}_f(\tau)]_k=e^{-j2\pi k\Delta f\tau}$ and $[\vec{a}_t(\nu)]_n=e^{j2\pi\nu nT_{\rm s}}$.

\subsection{Problem Formulation and Proposed Solutions}
\label{subsec:problem}

Let $\mathcal Z=\{Z_m[k,n]\}_{m=1}^{M}$ collect independent, phase-aligned grids over which $H_{\rm t}$, $H_{\rm c}$, and $H_{\rm I}$ are constant, while $D_m$ and $W_m$ are independent across $m$. Each cell in (\ref{eq:z}) contains a desired target term, deterministic clutter, a randomly modulated co-channel interference term, and Gaussian noise. Consistent with this sensing-centric scope, the first task at the BS is
\begin{equation}
\mathcal H_0:\mathbf{H}_{\rm I}=\mathbf{0}
\quad\text{versus}\quad
\mathcal H_1:\mathbf{H}_{\rm I}\neq\mathbf{0}.                              \label{eq:problem_hyp}
\end{equation}
The cellwise powers are $p_{\rm I}[k,n]=|H_{\rm I}[k,n]|^2/P_t$ and $p_w[k,n]=\mathbb{E}\big[|\widetilde W[k,n]|^2\big]$, while target detection is enhanced subject to $\varepsilon(w)\geq\varepsilon_0$, where $\varepsilon(w)$ denotes the effective aperture efficiency and $\varepsilon_0$ is its prescribed floor. Thus, the required outputs are an interference-presence decision, cellwise interference and noise profiles, and a range-Doppler detector that does not sacrifice excessive aperture efficiency.

The proposed solution begins after reference removal, as emphasized in Fig.~\ref{fig:problem_chain}, and treats the two disturbances separately: moving-target indication (MTI) rejects the deterministic zero-Doppler clutter, whereas centered diagonal fourth-order statistics identify the non-Gaussian co-channel component. Bias correction and a second-order moment then separate interference from noise, and diagonally loaded disturbance weights feed the two-dimensional discrete Fourier transform (DFT) range-Doppler map and a constant false-alarm rate (CFAR) detector. Sections~\ref{sec:proposed} and~\ref{sec:weighting} develop these steps.

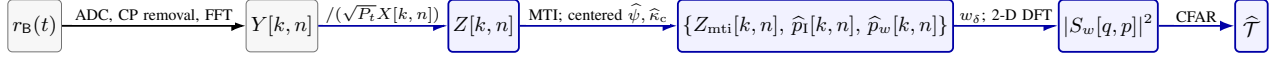
\begin{figure*}[t]
\centering
\resizebox{0.92\textwidth}{!}{%
\begin{tikzpicture}[
  >=Latex,
  font=\small,
  standard/.style={draw=black!55, rounded corners=2pt, align=center,
    minimum height=8mm, minimum width=17mm, fill=black!3, inner sep=2pt},
  proposed/.style={draw=blue!65!black, thick, rounded corners=2pt,
    align=center, minimum height=8mm, fill=blue!5, inner sep=2pt},
  edge text/.style={font=\scriptsize, align=center, fill=white, inner sep=1pt},
  flow/.style={-{Latex[length=2.2mm]}, thick},
  propflow/.style={-{Latex[length=2.2mm]}, thick, draw=blue!65!black}
]
  \node[standard, minimum width=5mm] (echo) at (-0.5,0) {$r_{\mathsf B}(t)$};
  \node[standard, minimum width=5mm] (y) at (3.3,0) {$Y[k,n]$};
  \node[proposed, minimum width=5mm] (z) at (6.4,0) {$Z[k,n]$};
  \node[proposed, minimum width=40mm] (est) at (11.5,0)
    {$\{Z_{\rm mti}[k,n],\,\widehat p_{\rm I}[k,n],\,\widehat p_w[k,n]\}$};
  \node[proposed, minimum width=5mm] (rd) at (16,0) {$|S_w[q,p]|^2$};
  \node[proposed, minimum width=5mm] (dec) at (18.2,0)
    {$\widehat{\mathcal T}$};

  \draw[flow] (echo) -- node[edge text, above]
    {ADC, CP removal, FFT} (y);
  \draw[propflow] (y) -- node[edge text, above]
    {$/(\sqrt{P_t}X[k,n])$} (z);
  \draw[propflow] (z) -- node[edge text, above]
    {MTI; centered $\widehat\psi,\widehat\kappa_{\rm c}$} (est);
  \draw[propflow] (est) -- node[edge text, above]
    {$w_\delta$; 2-D DFT} (rd);
  \draw[propflow] (rd) -- node[edge text, above]
    {CFAR} (dec);
\end{tikzpicture}%
}
\caption{BS sensing chain; proposed processing is shown in blue.}
\label{fig:problem_chain}
\end{figure*}

\subsection{Structural Distinction}

In (\ref{eq:rankone}) the target matrix is a sum of $L$ outer products, and static clutter collapses to a single one because $\vec{a}_t(0)=\vec{1}_N$, giving rank one irrespective of $G$; this is exact only at $\nu_g^{\rm c}=0$. The residual $\mathbf{H}_{\rm I}\odot\mathbf{U}$ inherits no such structure, since the independently modulated $\mathbf{U}$ destroys the outer-product form and leaves a generically full-rank matrix. For the finite-alphabet quadrature amplitude modulation (QAM) or phase-shift keying (PSK) draws used here, rank $\min(K,N)$ is verified numerically in Section~\ref{sec:results}. Known-waveform echoes are thus deterministic and cumulant-free while the interference residual is random and non-Gaussian, so MTI and fourth-order processing address them separately. We assume a constant-modulus reference $|X[k,n]|=1$ (QPSK in simulations); Remark~\ref{rem:cm} shows why a nonconstant modulus breaks fourth-order identification.

\section{Proposed Dual-Branch Processing}
\label{sec:proposed}

\subsection{MTI Clutter Rejection}

Static clutter is removed by slow-time mean cancellation,
\begin{equation}
Z_{\rm mti}[k,n]=Z[k,n]-\frac{1}{N}\sum_{n'=0}^{N-1}Z[k,n'].          \label{eq:mti}
\end{equation}

\begin{lemma}[MTI response]
\label{lem:mti}
For a slow-time tone $\vec{a}_t(\nu)$, the normalized energy retained by (\ref{eq:mti}) is
\begin{equation}
g_{\rm mti}^2(\nu)=1-\left|
\frac{\sin(\pi\nu NT_{\rm s})}{N\sin(\pi\nu T_{\rm s})}
\right|^2,                                                           \label{eq:mti_gain}
\end{equation}
read by continuity at $\nu T_{\rm s}\in\mathbb{Z}$, where $g_{\rm mti}=0$. Near zero,
\begin{equation}
g_{\rm mti}^2(\nu)=\frac{N^2-1}{3}(\pi\nu T_{\rm s})^2+O\!\left((\nu T_{\rm s})^4\right). \label{eq:mti_taylor}
\end{equation}
\end{lemma}
\begin{proof}
Mean removal is the orthogonal projector $\boldsymbol{\Pi}_\perp=\mathbf{I}_N-\vec{1}_N\vec{1}_N^T/N$. Since $\|\vec{a}_t(\nu)\|^2=N$, $\|\boldsymbol{\Pi}_\perp\vec{a}_t(\nu)\|^2/N=1-|\vec{1}_N^T\vec{a}_t(\nu)|^2/N^2$.
Evaluating the geometric sum gives (\ref{eq:mti_gain}); its Taylor expansion gives (\ref{eq:mti_taylor}).
\end{proof}
Thus, static clutter is nulled exactly; near-zero-Doppler clutter and targets are attenuated only for $|v|\ll\lambda/(2NT_{\rm s})$; at any nonzero Doppler bin $\nu=p/(NT_{\rm s})$, $1\leq p\leq N-1$, the Dirichlet numerator nulls and $g_{\rm mti}=1$, so the loss vanishes there. Higher-order processing does not replace MTI because deterministic clutter has zero fourth-order cumulant.

The unweighted range-Doppler response $S_0[q,p]=\allowbreak\sum_{k=0}^{K-1}\sum_{n=0}^{N-1}\allowbreak Z_{\rm mti}[k,n]\allowbreak e^{j2\pi kq/K}\allowbreak e^{-j2\pi np/N}$ completes the baseline branch of Fig.~\ref{fig:problem_chain}; maps are displayed in dB relative to the peak. Its axes are $d_q=cq/(2K\Delta f)$ and $v_p=-\lambda\widetilde p/(2NT_{\rm s})$ \cite{Richards2014}, where $\widetilde p=p$ for $p<\lceil N/2\rceil$ and $\widetilde p=p-N$ otherwise. Section~\ref{sec:weighting} replaces the uniform aperture in $S_0$ by interference-aware weights.

\subsection{Fourth-Order Interference Identification}

Let $z_m[k,n]$, $m=1,\ldots,M$, denote the phase-aligned repeated frames of Section~\ref{subsec:problem}. For complex $z$, let $\operatorname{cum}_4(z)\triangleq\operatorname{cum}(z,z^*,z,z^*)$ denote the fourth-order cumulant with two conjugated arguments. It is translation invariant and, for zero-mean $z$, evaluates to
\begin{equation}
\operatorname{cum}_4(z)=\mathbb{E}\big[|z|^4\big]-2\Big(\mathbb{E}\big[|z|^2\big]\Big)^2-\big|\mathbb{E}\big[z^2\big]\big|^2.  \label{eq:cumdef}
\end{equation}
The last term vanishes in population for QPSK and square QAM but is retained because its finite-sample expectation is nonzero: the additive constant in Proposition~\ref{prop:bias} is $4/M$ with it and $2/M$ without.

\begin{lemma}[Fourth-order selectivity]
\label{lem:cumulant}
Assume a constant-modulus reference $|X[k,n]|=1$, deterministic echoes over the estimation window, proper independent $D[k,n]$, and proper Gaussian $W[k,n]$. At every resource element,
\begin{equation}
\kappa[k,n]=\operatorname{cum}_4(Z[k,n])
=\frac{\kappa_U|H_{\rm I}[k,n]|^4}{P_t^2},                           \label{eq:diag}
\end{equation}
where $\kappa_U=\operatorname{cum}_4(U)$, equal to $-1$ for proper PSK and $-0.68$ for normalized 16-QAM. The value is invariant to target and clutter amplitudes and to Gaussian-noise power.
\end{lemma}
\begin{proof}
At a fixed cell, write $Z=c+aU+\widetilde W$, where $c=H_{\rm t}+H_{\rm c}$ and $a=H_{\rm I}/\sqrt{P_t}$. Because $|X|=1$, $U=D/X=DX^{*}$ is a phase rotation of $D$, so $\operatorname{cum}_4(U)=\operatorname{cum}_4(D)$, and $\widetilde W=WX^{*}/\sqrt{P_t}$ is proper Gaussian with unchanged variance; circular symmetry of $W$ makes $\widetilde W$ independent of $U$ even though both carry $X^{*}$. Cumulants of order at least two are translation invariant, cumulants of independent sums add, and Gaussian cumulants above order two vanish \cite{MendelProcIEEE1991}. Because two of the four arguments are conjugated, homogeneity scales the statistic by $a\,a^*a\,a^*=|a|^4$ rather than $a^4$, giving $\operatorname{cum}_4(Z)=|a|^4\operatorname{cum}_4(U)$, which is (\ref{eq:diag}) and is real for complex $a$.
\end{proof}

\begin{remark}[Necessity of a constant-modulus reference]
\label{rem:cm}
If $|X|$ is not constant, $U=D/X$ and $\widetilde W=W/(\sqrt{P_t}X)$ share $X$ and are dependent, so cumulants no longer add. With $\eta\triangleq\mathbb{E}\big[|X|^{-2}\big]$ and $p_w=\sigma_w^2/P_t$ the pre-division noise power,
$\operatorname{cum}_4(Z)=\big[p_{\rm I}^{2}\mathbb{E}\big[|D|^{4}\big]+4p_{\rm I}p_w+2p_w^{2}\big]\mathbb{E}\big[|X|^{-4}\big]-2\big(p_{\rm I}+p_w\big)^{2}\eta^{2}$,
which reduces to (\ref{eq:diag}) only when $\mathbb{E}\big[|X|^{-4}\big]=\eta^2=1$. Otherwise cross terms proportional to $\mathbb{E}[|X|^{-4}]-\eta^2$ survive, (\ref{eq:diag}) is no longer invertible for $p_{\rm I}$ alone, and the polarity can even invert: a 16-QAM reference with a 16-QAM interferer gives $\operatorname{cum}_4(D/X)>0$.
\end{remark}

The statistic is formed on $Z$ before MTI to avoid scaling and mixing the interference profile. Under symbol independence, cross-cell fourth-order cumulants vanish; hence only the $KN$ diagonal statistics in (\ref{eq:diag}) are required. This reduces storage from $O((KN)^4)$ for a full fourth-order tensor to $O(KN)$, with $O(MKN)$ arithmetic for the proposed cellwise estimator.

\subsubsection{Estimator and Finite-Sample Correction}

Define the centered sample as $\dot z_m=z_m-\bar z$, where $\bar z=M^{-1}\sum_m z_m$. The associated moments are $\widehat m_2=M^{-1}\sum_m|\dot z_m|^2$ and $\widehat m_{20}=M^{-1}\sum_m\dot z_m^2$.
The natural estimator is
\begin{equation}
\widehat\kappa=M^{-1}\sum_m|\dot z_m|^4-2\widehat m_2^2-|\widehat m_{20}|^2. \label{eq:khat}
\end{equation}
Centering removes the deterministic target and clutter but creates a finite-sample cumulant bias.

\begin{proposition}[Finite-sample bias of the centered cumulant]
\label{prop:bias}
For $M\geq4$, let $z_m=c+e_m$ with deterministic $c$ and independent zero-mean proper $e_m$ satisfying $\mathbb{E}\big[|e_m|^2\big]=\psi$ and $\operatorname{cum}_4(e_m)=\kappa$. Then
\begin{align}
\mathbb{E}\big[\widehat\kappa\big]&=-\frac{4\psi^2(M-1)}{M^2}+A_M\kappa,      \nonumber\\
\mathbb{E}\big[\widehat m_2^2\big]&=\frac{\psi^2(M-1)}{M}+\frac{(M-1)^2}{M^3}\kappa, \label{eq:bias}
\end{align}
with $A_M=\big[(M-1)^4+(M-1)\big]M^{-4}-3(M-1)^2M^{-3}$. Centering therefore biases $\widehat\kappa$ in two distinct ways: an additive term that is present even when $\kappa=0$, and a multiplicative shrinkage of $\kappa$ itself. The additive term is removed by $\widehat\kappa+4M^{-1}\widehat m_2^2$, which has exactly zero mean under the Gaussian null $\kappa=0$ but satisfies $\mathbb{E}\big[\widehat\kappa+4M^{-1}\widehat m_2^2\big]=B_M\kappa$ in general, where
\begin{equation}
B_M=\frac{(M-1)(M-2)(M^2-4M+2)}{M^4}.                                 \label{eq:BM}
\end{equation}
The fully corrected estimator
\begin{equation}
\widehat\kappa_{\rm c}=\frac{1}{B_M}
\left(\widehat\kappa+\frac{4}{M}\widehat m_2^2\right)                 \label{eq:corrected}
\end{equation}
is thus unbiased for every proper disturbance of finite fourth moment, Gaussian or not. Propriety is essential rather than cosmetic: improper laws contribute additional $\mathbb{E}\big[e^2\big]$ pairings that alter both $A_M$ and $B_M$.
\end{proposition}
\begin{proof}
Let $\boldsymbol{\Pi}=\mathbf{I}_M-\vec{1}_M\vec{1}_M^T/M$ and $\dot{\vec{z}}=\boldsymbol{\Pi}\vec{e}$. For a linear form $y=\sum_jc_je_j$ of independent proper variables, $\mathbb{E}\big[|y|^2\big]=\psi\sum_j|c_j|^2$, $\operatorname{cum}_4(y)=\kappa\sum_j|c_j|^4$ and $\mathbb{E}\big[|y|^4\big]=\operatorname{cum}_4(y)+2\Big(\mathbb{E}\big[|y|^2\big]\Big)^2$. Applying this to the rows of $\boldsymbol{\Pi}$, for which $\sum_j\Pi_{mj}^2=1-M^{-1}$ and $\sum_j\Pi_{mj}^4=[(M-1)^4+(M-1)]M^{-4}$, gives the first moment. For the quadratic forms, independence and propriety yield $\mathbb{E}\big[|\vec{e}^H\boldsymbol{\Pi}\vec{e}|^2\big]=\psi^2\big[(\operatorname{tr}\boldsymbol{\Pi})^2+\operatorname{tr}\boldsymbol{\Pi}^2\big]+\kappa\sum_i\Pi_{ii}^2$ and $\mathbb{E}\big[|\vec{e}^T\boldsymbol{\Pi}\vec{e}|^2\big]=2\psi^2\operatorname{tr}\boldsymbol{\Pi}^2+\kappa\sum_i\Pi_{ii}^2$,
with $\operatorname{tr}\boldsymbol{\Pi}=\operatorname{tr}\boldsymbol{\Pi}^2=M-1$ and $\sum_i\Pi_{ii}^2=(M-1)^2/M$. Substituting in (\ref{eq:khat}) gives (\ref{eq:bias}); collecting the $\kappa$ terms of $\widehat\kappa+4M^{-1}\widehat m_2^2$ and factoring yields (\ref{eq:BM}).
\end{proof}

Both corrections matter, and in opposite directions. Because standard QAM/PSK has $\kappa_U<0$, the additive bias is negative and mimics interference, whereas the shrinkage $A_M<1$ understates it; at $M=96$, $A_M=0.93$ and the two errors largely cancel near the reference SIR, so an estimator that applies neither correction can appear accurate for the wrong reason. Applying (\ref{eq:corrected}) removes both. Fourth-order error still converges as $O(M^{-1/2})$, so coherence time determines usable snapshot support.

\subsubsection{Power Decomposition and Presence Test}

For constant-modulus $X$ and unit-power $U$, the centered second moment is
\begin{equation}
\psi[k,n]=\frac{|H_{\rm I}[k,n]|^2}{P_t}+\frac{\sigma_w^2}{P_t}.      \label{eq:psi}
\end{equation}
Its unbiased estimate is $\widehat\psi=M\widehat m_2/(M-1)$.
Pairing (\ref{eq:diag}) and (\ref{eq:psi}) gives
\begin{align}
\widehat p_{\rm I}[k,n]&=\left[\frac{\widehat\kappa_{\rm c}[k,n]}{\kappa_U}\right]_+^{1/2}, \label{eq:pI}\\
\widehat p_w[k,n]&=\left[\widehat\psi[k,n]-\widehat p_{\rm I}[k,n]\right]_+,               \label{eq:pw}
\end{align}
where $[x]_+=\max(x,0)$ and the powers are normalized by $P_t$. The two clips are applied sequentially, so $\widehat p_{\rm I}+\widehat p_w=\widehat\psi$ holds whenever $\widehat p_{\rm I}\leq\widehat\psi$, which covers every operating point reported here. Second order supplies only their sum $\psi$; fourth order supplies the split. This distinction permits (\ref{eq:problem_hyp}) to be tested without knowing $\sigma_w^2$: the grid-averaged corrected cumulant has zero mean under $\mathcal H_0$ and a strictly negative mean under $\mathcal H_1$ for QAM/PSK, and normalizing it by $\widehat\psi^2$ makes its null law free of the disturbance scale. Calibration on interference-free frames then sets a one-sided threshold at the required false-alarm probability.

Proposition~\ref{prop:bias} makes $\widehat\kappa_{\rm c}$ unbiased, but (\ref{eq:pI})-(\ref{eq:pw}) apply a square root and a clip, so under $\mathcal H_0$, $\widehat p_{\rm I}$ has an $O(M^{-1/4})$ positive bias. With $\psi=1$ we measure $\mathbb{E}\big[\widehat p_{\rm I}\big]=0.23$ at $M=96$ and $0.16$ at $M=384$, with a matching negative bias in $\widehat p_w$. The presence test is therefore stated on the grid-averaged $\widehat\kappa_{\rm c}$, where the nonlinearity is absent and Proposition~\ref{prop:bias} applies exactly, rather than on the per-cell $\widehat p_{\rm I}$.

Fig.~\ref{fig:fourth} confirms that the recovered profile tracks the selective interference channel, that the second-order moment follows interference plus noise, and that the error decays as $O(M^{-1/2})$; echo invariance is checked in Section~\ref{sec:results}. The construction assumes proper disturbances, a constant-modulus reference with $\kappa_U=\operatorname{cum}_4(D)$ known for (\ref{eq:pI}) and only its sign needed for the presence test (Remark~\ref{rem:cm}), and $H_{\rm t}$, $H_{\rm c}$, and $H_{\rm I}$ constant over the $M$ snapshots.

\begin{figure*}[t]
\centering
\includegraphics[width=\textwidth]{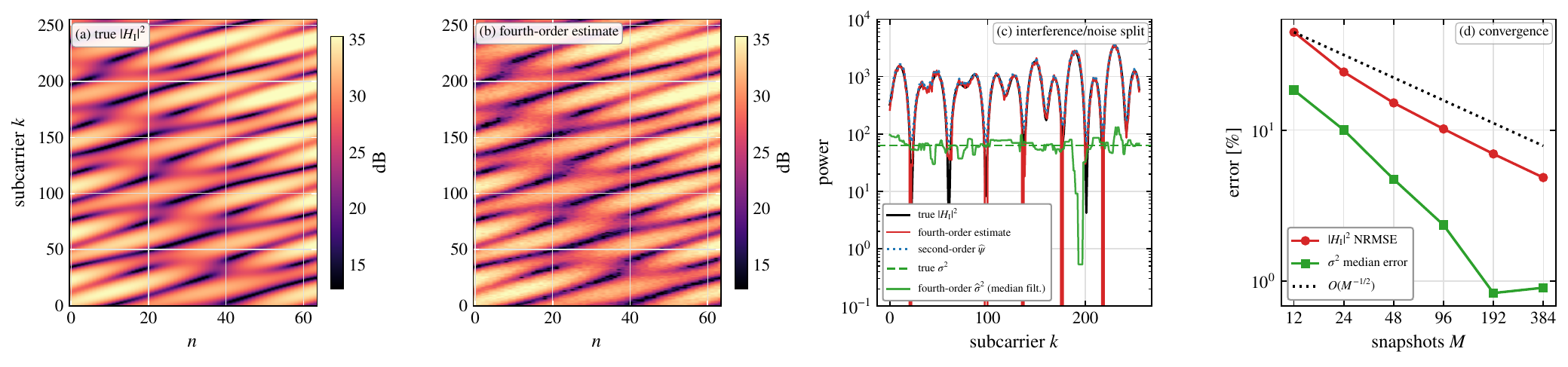}
\caption{Fourth-order validation for 16-QAM interference: (a) true and (b) estimated interference power at $M=96$; (c) interference/noise separation; (d) interference normalized root-mean-square error (NRMSE) and noise median error versus $M$.}
\label{fig:fourth}
\end{figure*}

\section{Interference-Aware Detection}
\label{sec:weighting}

\subsection{Weighted Range-Doppler Formation}

The weighted range-Doppler response extends $S_0$ as $S_w[q,p]=\allowbreak\sum_{k,n}\allowbreak w[k,n]Z_{\rm mti}[k,n]\allowbreak e^{j2\pi kq/K}\allowbreak e^{-j2\pi np/N}$, where $w[k,n]=1$ recovers $S_0$.

\begin{proposition}[Inverse-power weighting gain]
\label{prop:gain}
Approximate the post-MTI disturbance covariance by $\operatorname{diag}(\psi_1,\ldots,\psi_P)$, with $\psi_i>0$ and $P=KN$. For a matched on-grid target, $w_i\propto1/\psi_i$ maximizes the coherent output signal-to-noise ratio (SNR), and its gain over uniform weighting is
\begin{equation}
G_{\rm w}=\frac{1}{P^2}\left(\sum_{i=1}^{P}\psi_i\right)
\left(\sum_{i=1}^{P}\frac{1}{\psi_i}\right)\geq1,                   \label{eq:gain}
\end{equation}
with equality only for constant $\psi_i$.
\end{proposition}
\begin{proof}
After steering-phase removal, the target amplitude is $a$ in every cell, so the output SNR is $|a|^2(\sum_iw_i)^2/\sum_iw_i^2\psi_i$. Cauchy--Schwarz gives $w_i\propto1/\psi_i$ and (\ref{eq:gain}).
\end{proof}
MTI induces slow-time correlation; hence (\ref{eq:gain}) is the diagonal-covariance prediction for on-grid targets, using the pre-MTI powers as a proxy. The unbiased $\widehat\psi=M\widehat m_2/(M-1)$ is second order, so fourth order is unnecessary for weighting. In population, $p_{\rm I}+p_w=\psi$; the clipped estimates give the same weights when $\widehat p_{\rm I}\leq\widehat\psi$ and agree numerically here. Fourth order instead provides identification, power decomposition, and presence testing.

\subsection{Aperture-Controlled Diagonal Loading}

Unregularized $1/\widehat\psi$ weighting is nevertheless unsafe. It is a nonuniform aperture taper, and for $w_i>0$ its efficiency $\varepsilon(w)=(\sum_iw_i)^2/(P\sum_iw_i^2)\in[1/P,1]$ can collapse as the disturbance dynamic range grows. For white disturbance with variance $\psi$, the weighted coherent output SNR is $\mathrm{SNR}_{w}=|a|^2(\sum_i w_i)^2/(\psi\sum_i w_i^2)$, while uniform weighting gives $\mathrm{SNR}_{\rm unif}=P|a|^2/\psi$. Therefore, $\mathrm{SNR}_{w}/\mathrm{SNR}_{\rm unif}=\varepsilon(w)$, and $-10\log_{10}\varepsilon(w)$ is the corresponding coherent-processing/taper loss in dB. In the selective-disturbance case, (\ref{eq:gain}) quantifies the interference-suppression gain, whereas $\varepsilon(w)$ quantifies the associated aperture/taper penalty. By Parseval, $\varepsilon$ is exactly the fraction of point-spread-function (PSF) energy in the peak bin, so low efficiency raises the PSF sidelobe floor and increases leakage into the CFAR training ring; transmit waveform design addresses the same sidelobe budget \cite{LiTWC2025}, whereas we act at the receiver. We therefore use $w_{\delta,i}=1/(\widehat\psi_i+\delta\bar\psi)$, where $\bar\psi=P^{-1}\sum_i\widehat\psi_i$ and, for $1/P\leq\varepsilon_0<1$, the smallest $\delta\geq0$ satisfying $\varepsilon(w_\delta)\geq\varepsilon_0$ is found by bisection on $\varepsilon(w_\delta)$, which increases numerically with $\delta$ and tends to one as $\delta\to\infty$. This loading retains most of the coherent gain while improving detection relative to uniform weighting (Fig.~\ref{fig:tradeoff}). Guard cells are sized using the loaded, not uniform, PSF.

\subsection{CFAR Detection}

Detection uses cell-averaging CFAR (CA-CFAR) processing on $|S_w[q,p]|^2$. With $N_T$ exponential training cells and desired false-alarm probability $P_{\rm FA}$, its threshold scale is $\alpha_{\rm CFAR}=N_T\big(P_{\rm FA}^{-1/N_T}-1\big)$.
Training cells are exactly exponential and independent only for a white circular Gaussian disturbance; after MTI, cellwise weighting, and the DFT of a QAM residual this holds only approximately, so $\alpha_{\rm CFAR}$ is used as a design scale and the realized false-alarm rate is reported in Section~\ref{sec:results}. Ordered-statistic CFAR (OS-CFAR) is included as a nonhomogeneous-background benchmark \cite{Barkat2005,Richards2014}. Thus the complete chain is reference removal $\rightarrow$ MTI $\rightarrow$ interference identification $\rightarrow$ loaded range-Doppler processing $\rightarrow$ CFAR.

\begin{figure}[t]
\centering
\includegraphics[width=0.98\columnwidth]{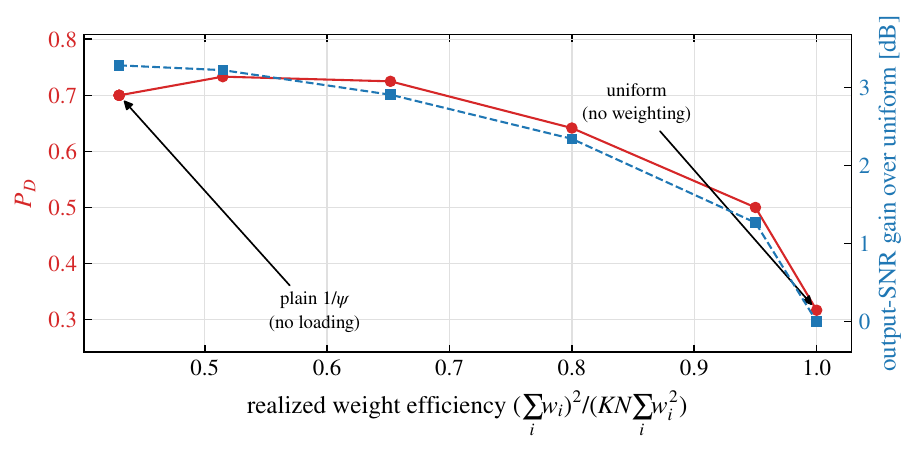}
\caption{Loading tradeoff among output gain, aperture efficiency, and detection probability $P_D$.}
\label{fig:tradeoff}
\end{figure}

\section{Numerical Results}
\label{sec:results}

We simulate the model in Section~\ref{sec:model} using the configuration in Table~\ref{tab:config}. The scenario contains off-grid moving targets, static clutter, and a frequency- and time-selective OFDM interferer; Fig.~\ref{fig:topology_geom} shows the node geometry. The per-resource-element SNR, signal-to-interference ratio (SIR), and signal-to-clutter ratio (SCR) are defined relative to a target echo; each experiment sweeps the relevant ratio and otherwise holds the reference point of Table~\ref{tab:config}. The CA-CFAR branches share one design, with OS-CFAR as a benchmark.

\begin{table}[t]
\caption{Simulation parameters: settings (left) and quantities derived from them (right), using $c=3\times10^8$~m/s.}
\label{tab:config}
\centering
\renewcommand{\arraystretch}{1.06}
\begin{tabular}{lc@{\hspace{1.3em}}lc}
\toprule
Parameter & Value & Parameter & Value\\
\midrule
$f_c$ & 28 GHz & $\lambda$ & 10.7 mm\\
$\Delta f$ & 120 kHz & $B$ & 30.72 MHz\\
$K,N$ & $256,64$ & $KN$ & 16384\\
$T_{\rm cp}/T$ & $1/8$ & CPI & 0.600 ms\\
Reference & QPSK & Range res. & 4.883 m\\
Interferer & 16-QAM & Velocity res. & 8.929 m/s\\
Interferer taps & 4 & CP range limit & 156.2 m\\
CFAR window & $21\times21$ & $N_T$ & 360\\
CFAR guard & $9\times9$ & $\alpha_{\rm CFAR}$ & 9.33\\
$P_{\rm FA}$ & $10^{-4}$ & $\eta$ & 1\\
$M,\varepsilon_0$ & $96,\,0.5$ & $\kappa_U$ & $-0.68$\\
SNR, SIR, SCR & $-18,-30,-20$ dB & $L,G$ & $2,6$\\
\bottomrule
\end{tabular}
\end{table}

\begin{figure}[t]
\centering
\includegraphics[width=0.8\columnwidth]{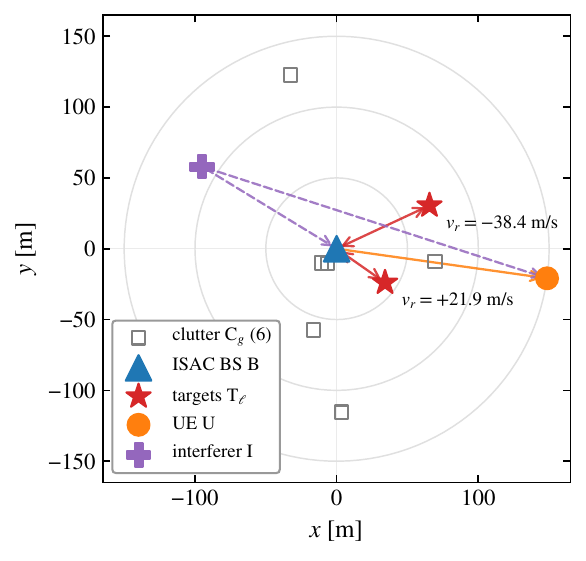}
\caption{Simulated geometry with two off-grid targets, six static clutter scatterers, and a co-channel interferer.}
\label{fig:topology_geom}
\end{figure}

\begin{figure*}[!t]
\centering
\includegraphics[width=\textwidth]{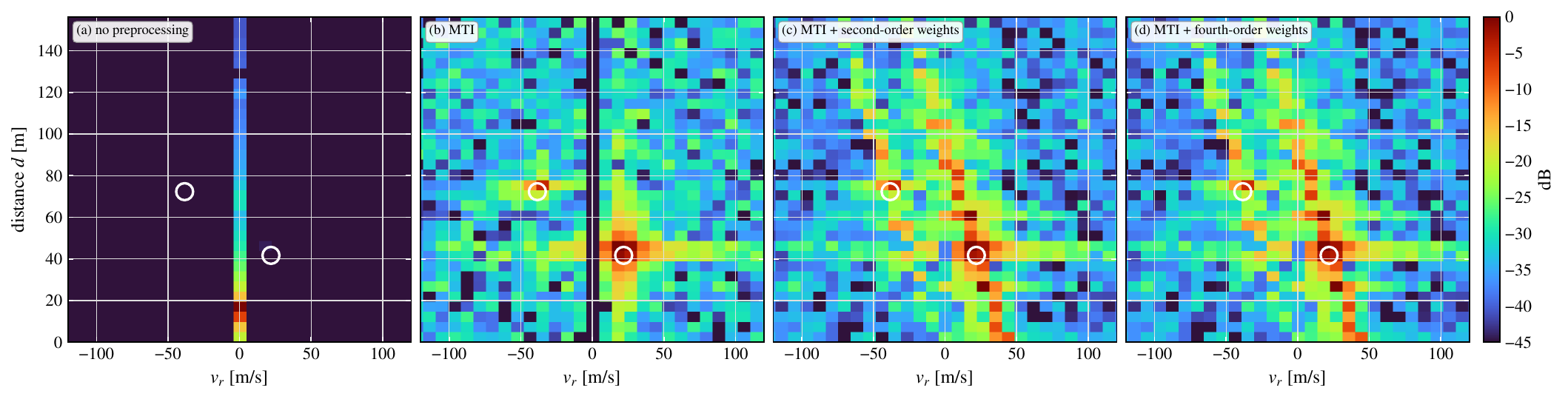}
\caption{Range-Doppler maps at SNR $=6$ dB, SIR $=-6$ dB, and SCR $=-32$ dB: (a) unprocessed; (b) MTI; (c) second-order weighting; (d) fourth-order weighting. Circles mark targets.}
\label{fig:rdmaps}
\end{figure*}

\begin{figure}[!t]
\centering
\subfloat[]{\includegraphics[width=\columnwidth]{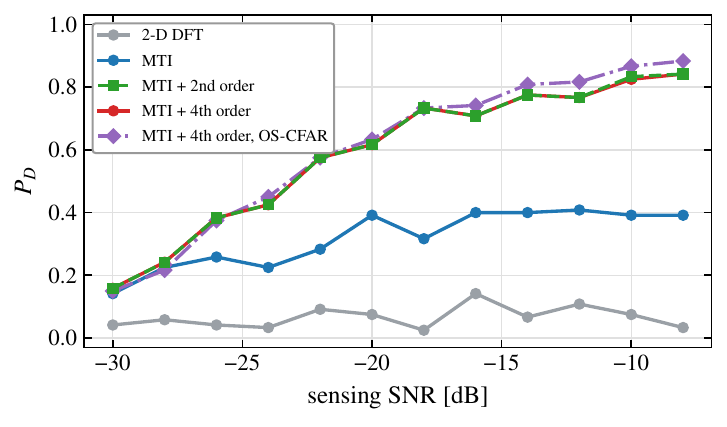}\label{fig:pd_snr}}\\[-0.5ex]
\subfloat[]{\includegraphics[width=\columnwidth]{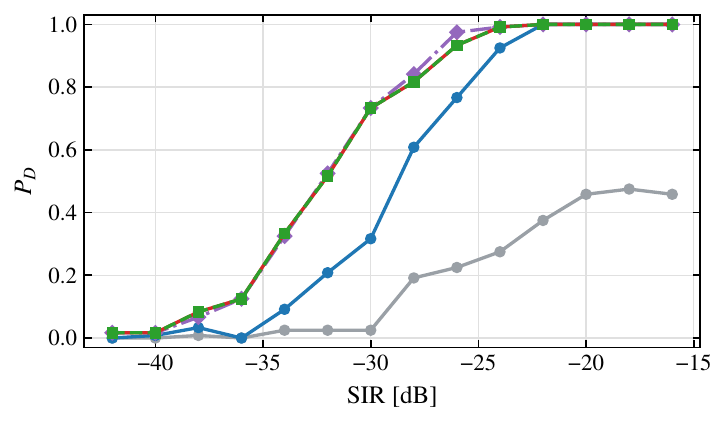}\label{fig:pd_sir}}
\caption{Detection probability versus (a) SNR at SIR $=-30$ dB and (b) SIR at SNR $=-18$ dB. All branches use CFAR.}
\label{fig:pd}
\end{figure}

The rank experiment confirms the structural distinction of Section~\ref{sec:model}: the interference, clutter, and target components have effective ranks $\min(K,N)$, one, and $L$. The measured MTI response follows Lemma~\ref{lem:mti}, reaching a numerical null for static clutter and attenuating sufficiently slow targets as predicted.

The maps in Fig.~\ref{fig:rdmaps} give a complementary view at a higher-contrast operating point. Before processing, the zero-Doppler clutter ridge and its sidelobes dominate. MTI removes that ridge without estimating clutter delays but leaves the spread co-channel component; target contrast improves only once the selective disturbance profile is used, confirming that the two stages are non-redundant.

The fourth-order experiments validate Lemma~\ref{lem:cumulant}: the recovered interference profile approaches the theoretical profile at the predicted $O(M^{-1/2})$ rate, while scaling the deterministic target and clutter amplitudes leaves the cumulant profile unchanged to numerical precision, as centering removes snapshot-constant echoes exactly. The additive step of Proposition~\ref{prop:bias} drives the normalized Gaussian-null mean to zero within Monte Carlo error, preventing finite-sample bias from mimicking interference, and omitting the shrinkage $B_M$ then inflates the noise median error by more than an order of magnitude; applying neither correction is accurate only by cancellation rather than by construction.

The interference-presence test (\ref{eq:problem_hyp}) is calibrated on interference-free frames. Averaging the corrected cumulant over the grid and normalizing by $\widehat\psi^{2}$ leaves a null law that is Gaussian and free of the disturbance scale, so a single threshold serves every noise level and no knowledge of $\sigma_w^2$ is needed. At the design false-alarm rate the test then declares interference far below the reference operating point of Table~\ref{tab:config}, and in particular well below the weakest interference at which weighting still changes $P_D$; its power falls off only once $T\propto(p_{\rm I}/p_w)^{2}$ sinks beneath the null spread. Second order cannot make this decision at any $I/N$, since $\widehat\psi$ measures total disturbance power without attributing it, which is the operational content of the fourth-order statistic.

Fig.~\ref{fig:pd} compares the CA-CFAR branches with no suppression, MTI alone, and estimated disturbance weighting. Loaded weighting more than doubles $P_D$ over MTI alone at the reference point, and the improvement persists as SNR and SIR are swept. The second- and fourth-order weight curves overlap, as predicted by (\ref{eq:psi}); the distinctive fourth-order benefit is interference/noise decomposition and presence testing. The realized false-alarm rate tracks the $P_{\rm FA}$ design value across all branches and SNR points.

Fig.~\ref{fig:tradeoff} isolates the PSF effect at the reference point, reusing Fig.~\ref{fig:pd}'s seed base and trial count so the two agree exactly rather than within Monte Carlo variance. Diagonal loading retains nearly all of the inverse-weighting gain while raising both aperture efficiency and $P_D$, so the headroom is almost free; forcing an efficiency floor near unity instead surrenders most of the gain and much of $P_D$.

\section{Conclusion}
\label{sec:conclusion}

In this paper, we developed a sensing-centric receiver for monostatic OFDM-ISAC in clutter and co-channel interference. Reference removal exposes complementary structures: deterministic zero-Doppler clutter and a random, non-Gaussian, generically full-rank interference residual. We therefore combined MTI with a centered diagonal fourth-order statistic, corrected its finite-sample bias for proper phase-aligned snapshots, and paired it with second-order information to separate interference from noise. Under a diagonal post-MTI covariance approximation, aperture-controlled loading retains most of the disturbance-suppression gain while enforcing an efficiency floor. Simulations confirm the predicted effective-rank distinction, show that interference estimates improve with coherent support, and demonstrate stronger target detection with loaded processing.

\bibliographystyle{IEEEtran}
\bibliography{references}

@misc{Triwidyastuti2026arxiv,
	title={{Communication-Centric RIS-Assisted ISAC: Signal Modeling and BER Analysis}},
	author={Yosefine Triwidyastuti and Tri Nhu Do},
	year={2026},
	eprint={2606.28924},
	archivePrefix={arXiv},
	primaryClass={eess.SP},
	url={https://arxiv.org/abs/2606.28924},
}

@misc{Nguyen_TWC_arXiv_2026,
  title         = "Probabilistic Denoising-enhanced {ISAC} for stochastic
                   cluttered mobile environments",
  author        = "Nguyen, Nghia Thinh and Do, Tri Nhu",
  month         =  jul,
  year          =  2026,
  archivePrefix = "arXiv",
  primaryClass  = "eess.SP",
  eprint        = "2607.26994",
    url={https://arxiv.org/abs/2607.26994}
}

@article{Thinh_TGCN_2026,
  author={Nghia Thinh Nguyen and Tri Nhu Do},
  journal={IEEE Trans. Green Commun. Netw.},
  title={{Generative and Explainable AI for High-Dimensional MIMO-OFDM Channel Estimation in Time-Frequency-Space Domain}},
  year={2026},
  volume={10},
  pages={3060--3075},
  doi={10.1109/TGCN.2026.3693617}
}

@article{MirabellaAccess2023,
  author={Michele Mirabella and Pasquale Di Viesti and Alessandro Davoli and Giorgio M. Vitetta},
  journal={IEEE Access},
  title={{Deterministic Signal Processing Techniques for OFDM-Based Radar Sensing: An Overview}},
  year={2023},
  volume={11},
  pages={68872--68889},
  doi={10.1109/ACCESS.2023.3292937}
}

@article{LuoTWC2024,
  author={Hongliang Luo and Yucong Wang and Dongqi Luo and Jianwei Zhao and Huihui Wu and Shaodan Ma and Feifei Gao},
  journal={IEEE Trans. Wireless Commun.},
  title={{Integrated Sensing and Communications in Clutter Environment}},
  year={2024},
  volume={23},
  number={9},
  pages={10941--10956},
  doi={10.1109/TWC.2024.3377184}
}

@article{MendelProcIEEE1991,
  author={Jerry M. Mendel},
  journal={Proc. IEEE},
  title={{Tutorial on Higher-Order Statistics (Spectra) in Signal Processing and System Theory: Theoretical Results and Some Applications}},
  year={1991},
  volume={79},
  number={3},
  pages={278--305},
  doi={10.1109/5.75086}
}

@BOOK{Barkat2005,
  title     = "Signal Detection and Estimation",
  author    = "Barkat, Mourad",
  publisher = "Artech House",
  series    = "Radar Library",
  edition   =  2,
  month     =  aug,
  year      =  2005,
  address   = "Norwood, MA",
}

@BOOK{Richards2014,
  title     = "Fundamentals of radar signal processing, second edition",
  author    = "Richards, Mark A",
  publisher = "McGraw-Hill Professional",
  edition   =  2,
  month     =  jan,
  year      =  2014,
  address   = "New York, NY",
  language  = "en"
}

@article{MengTWC2024,
  author={Kaitao Meng and Christos Masouros and Guangji Chen and Fan Liu},
  journal={IEEE Trans. Wireless Commun.},
  title={{Network-Level Integrated Sensing and Communication: Interference Management and BS Coordination Using Stochastic Geometry}},
  year={2024},
  volume={23},
  number={12},
  pages={19365--19381},
  doi={10.1109/TWC.2024.3450228}
}

@article{LiTWC2025,
  author={Peishi Li and Ming Li and Rang Liu and Qian Liu and A. Lee Swindlehurst},
  journal={IEEE Trans. Wireless Commun.},
  title={{MIMO-OFDM ISAC Waveform Design for Range-Doppler Sidelobe Suppression}},
  year={2025},
  volume={24},
  number={2},
  pages={1001--1015},
  doi={10.1109/TWC.2024.3503605}
}

\end{document}